\documentclass{article}

\usepackage[preprint]{neurips_2020}

\usepackage{booktabs}       % professional-quality tables
\usepackage{amsfonts}       % blackboard math symbols
\usepackage{nicefrac}       % compact symbols for 1/2, etc.
\usepackage{microtype}      % microtypography
\usepackage{xcolor}
\usepackage{subfig}
\usepackage{algorithmic} 

\usepackage{algorithm} 
\usepackage{makecell}

\usepackage{chrkroer}
\usepackage[capitalise]{cleveref}
\newtheorem{prop}{Proposition}

\newtheorem{lemma}{Lemma}

\begin{document}

\title{Implementing Fairness Constraints in Markets Using Taxes and Subsidies}

\author{%
  Alexander Peysakhovich \\
  Meta AI Research
   \And
  Christian Kroer \\
   Meta Core Data Science \\
   Columbia University
   \And
  Nicolas Usunier \\
   Meta AI Research
  % \texttt{email} \\
  % \And
  % Coauthor \\
  % Affiliation \\
  % Address \\
  % \texttt{email} \\
  % \And
  % Coauthor \\
  % Affiliation \\
  % Address \\
  % \texttt{email} \\
}

\maketitle

\begin{abstract}
Fisher markets are those where buyers with budgets compete for scarce items, a natural model for many real world markets including online advertising. A market equilibrium is a set of prices and allocations of items such that supply meets demand. We show how market designers can use taxes or subsidies in Fisher markets to ensure that market equilibrium outcomes fall within certain constraints. We show how these taxes and subsidies can be computed even in an online setting where the market designer does not have access to private valuations. We adapt various types of fairness constraints proposed in existing literature to the market case and show who benefits and who loses from these constraints, as well as the extent to which properties of markets including Pareto optimality, envy-freeness, and incentive compatibility are preserved. We find that some prior discussed constraints have few guarantees in terms of who is made better or worse off by their imposition.
\end{abstract}

\section{Introduction}
A market solves the economic problem of ``who gets what and why'' \citep{roth2015gets}. Unfortunately, there is no guarantee that market outcomes will always align with other objectives of the market designer such as business needs, legal constraints, or notions of fairness. Recent work proposes the addition of `fairness constraints' into various allocation mechanisms \citep{celis2019toward,ilvento2020multi,chawla2022individual,celli2022parity,balseiro2021regularized}. We study how market designers can use price interventions to ensure that market outcomes satisfy arbitrary linear constraints. We consider prior proposed interventions and study them in a market setting, paying particular attention to both individual and aggregate level consequences. Importantly, we show that second-order effects in markets can mean that well intended interventions do not achieve their goals.

We focus on Fisher markets where budget constrained buyers compete for items that are in limited supply \citep{eisenberg1959consensus}. Given prices for items, demand is the result of buyers maximizing their utilities. An equilibrium in Fisher markets are prices and allocations such that supply equals demand.

The Fisher model allows us to abstract away from the details of how a market works and study equilibria as properties of demand and supply directly. For example, even though individual impressions in many real world advertising markets are allocated via paced auctions, the \textit{aggregate} outcome across all advertisers and impressions is that prices make supply of ad slots equal to advertiser demand for ad slots -- in other words, a Fisher market equilibrium \citep{conitzer2019pacing,conitzer2022multiplicative}. 

From the perspectives of buyers and market designers, Fisher equilibria are well understood. The allocations are Pareto-optimal and envy-free up to budgets (no buyer strictly prefers another buyer's bundle to their own bundle when budgets are equal) \citep{varian1974efficiency,budish2012multi}. When markets are `large' buyers have no incentives to lie about their valuations of items \citep{azevedo2018strategy,peysakhovich2019fair}. 

However, equilibrium-based allocations can have undesirable distributive properties from the point of view of group-level fairness. For example, recent work argues advertising markets may result in outcomes which may be `unfair' in various senses~\citep{ali2019discrimination,zhang2021measuring,imana2021auditing}. We ask how market designers can `adjust' market outcomes to fix these undesirable distributional properties. 
 
A common tool for market designers is price intervention \citep{pigou1924economics} -- taxing or subsidizing some market transactions. Our main technical contribution is to give designers a way to target these in Fisher markets. Given a family of linear constraints, we show to construct a set of price interventions -- taxes and subsidies for some purchases -- such that there are allocations which satisfy the constraints and are also market equilibria. We show how to do this both in a full information one-shot setting, as well an online, decentralized setting using an algorithm we call online price intervention calculation (OPIC).

Many constraint families proposed in the literature can be written as linear constraints. For each of these proposed constraint families we ask: how are market outcome properties affected by the imposition of these constraints? Does imposition of these constraints always satisfy certain goals? Who wins and who loses when these interventions are made? The results are not always intuitive and, because second-order effects abound in markets, do not always achieve the stated goals of the constraints.

Motivated by work on ad delivery \cite{ali2019discrimination}, we consider buyer parity constraints where items are split into groups $A$ and $B$ and each buyer in a set is required to have parity in allocation with respect to these groups \citep{ali2019discrimination,celis2019toward,ilvento2020multi,chawla2022individual,celli2022parity}. This constraint guarantees parity in exposure between the $A$ and $B$ groups, however it does not guarantee Pareto optimality even when item group welfare is taken into account. In addition, we show that this intervention can lead to the previously disadvantaged group receiving \textit{less} aggregate exposure than before.

Motivated by the work on job search of \cite{geyik2019fairness} we also consider the `transpose' of the above constraint - per item parity constraints. Buyers have group labels $A$ and $B$ and some set of items are each required to have equal exposure in the buyer groups. Here the goal is that items (e.g. recruiter impressions) have parity in exposure to buyers (e.g. resumes of two groups of individuals). Here, again, we lack Pareto optimality and the intervention can reduce the utility of a previously disadvantaged buyer group.

We additionally consider a `floor' constraint where we require that a subset of buyers have a minimum exposure to a subset of items. We find that this intervention gives outcomes that are Pareto optimal if item exposure utility is taken into account but not otherwise. Satisfying this constraint in market outcomes is equivalent to the market designer subsidizing constrained buyers' purchases of the protected items, and though this may sound uniformly welfare improving, we show that buyers in the constrained group can have their welfare reduced by the intervention.

We show that each of the above constraint families maintains the incentive properties of Fisher markets as long as the group of buyers on whom the constraint applies grows large with the market.

Finally, many of our counter-examples are `worst case' scenarios. We consider random markets with constraints and ask what an `average case' outcome might look like.

Overall, our contribution can be split into two parts: the first is how a market designer \textbf{could} implement general linear constraints. The second is a focus whether they \textbf{should} implement a particular sets of constraints. The main takeaway for practitioners is that often there are no guarantees that second-order effects from constraint imposition may not create more problems than the constraint solves. We are not arguing for \textbf{no} interventions, but rather that interventions should be made carefully and thoughtfully.

\section{Related Work}\label{sec:related}
There is large interest in centralized market allocation as a solution to the multi-unit assignment problem \citep{budish2012multi}. Here individuals report their valuations for items, a center computes an (approximate) equal-budget equilibrium allocation, and returns this allocation to the agents. This mechanism is known as competitive equilibrium from equal incomes (CEEI, \cite{varian1974efficiency}) and is used in the allocation of courses to university students though with more complex utility functions and computations than our Fisher case \citep{budish2013designing,budish2016course}. Recent work studies how general constraints can be added to CEEI allocations \citep{echenique2021constrained}. Our work adds to this literature by showing properties of specific constraints proposed in discussions of demographic fairness in allocation. We do not seek to prove general results about all possible constraint families. Instead, we focus on the more restrictive case of linear constraints in Fisher markets which allows us to derive stronger results--i.e. the convex program to determine optimal taxes/subsidies--than in general unit demand markets studied by \cite{echenique2021constrained}.

Market equilibria in Fisher markets and their relationship to convex programming are well studied \citep{eisenberg1959consensus,shmyrev2009algorithm,cole2017convex,kroer2019computing,cole2017convex,cole2018approximating,caragiannis2016unreasonable,murray2019robust,gao2020first}. Our work complements this existing work as we show that the convex program equivalence allows us to easily construct market interventions in the forms of taxes and subsidies. In addition, we use this equivalence to construct a provably convergent online algorithm.

A recent literature has begun to study differential outcomes in online systems \citep{ali2019discrimination,geyik2019fairness,zhang2021measuring,imana2021auditing} across categories of users (e.g. job ads across gender). Though a highly studied cause of these differential outcomes are biases in machine learning systems, there are also market forces that can cause such issues. For example, \citet{lambrecht2018algorithmic} studies STEM job ads competing in a real ad market and shows that ``younger women are a prized demographic and are more expensive to show ads to. An algorithm that simply optimizes cost-effectiveness in ad delivery will deliver ads that were intended to be gender neutral in an apparently discriminatory way, because of crowding out.'' Our work further highlights the importance of understanding market dynamics both in terms of how they cause disparate allocations as well as how various interventions behave.

% Other work has begun to study unintended consequences of fairness interventions \citep{fang2011theories,liu2018delayed,mouzannar2019fair,emelianov2022fair}. Much of this work focuses on the case of unfairness due to biases in classification... [blah blah]

% Recent work focuses on the case of fairness interventions when agents are strategic \citep{hardt2016strategic,hu2019disparate,fang2011theories,mouzannar2019fair,elianov2022fairs}. Here various interventions can `go wrong' because they affect incentives of agents. The market setting is an interesting \alex{blahblah}

% \nico{The positioning wrt to this literature is the crux of the acceptance/rejection to FAccT. Also look at signaling games}

A recent literature has looked at implementing various notions of fairness in single auctions \citep{celis2019toward,ilvento2020multi,chawla2022individual}, paced auctions \citep{celli2022parity} or online allocation problems \citep{balseiro2021regularized}. Our work complements this by showing that in markets similar implementation can be done via the use of the price system. In addition, these papers study the performance of their algorithms but bypass the questions such as second order market effects, who wins and who loses, and whether good incentive properties are preserved.

As discussed in \citet{fazelpour2022algorithmic}, fairness constraints are often designed through desirable properties of ideal states. However, when the current state is far from ideal, enforcing the constraints may not bring us closer to the ideal state. Similar points has been made in the context of classification \citep{corbett2017algorithmic,menon2018cost, kasy2021fairness,liu2018delayed,weber2022enforcing}, various `statistical discrimination' scenarios \citep{fang2011theories,mouzannar2019fair,emelianov2022fair}, and counterfactual or causal definitions of fairness \citep{kusner2017counterfactual,imai2020principal,nilforoshan2022causal}. Our results about parity and floor constraints complement this literature showing that interventions in markets can sometimes (note, not always!) make things worse than before. In the Appendix we provide a longer discussion of the various results in this literature, which context they apply in, and how they relate to our results.

\section{Fisher Markets}\label{fisher_markets}
We will work with a market where there are $n$ buyers (e.g. advertisers) and $m$ items (e.g. ad slots). Each item has supply $s_j$ which for the purposes of this section we take to be $1$. We assume that fractional allocations are allowed (these can also be thought of as randomized allocations). Fractional allocation ensures existence and polynomial time computability of equilibria. In practice numerically that in many Fisher markets with linear utility, almost all assignments are $1$ or $0$ \citep{kroer2019scalable}.

We let $x \in \R^{n\times m}_+$ be an allocation of items to buyers. We let $x_i \in \mathbb{R}^m_{+}$ be the bundle of goods assigned to buyer $i$ with $x_{ij}$ being the assignment of item $j$. Each buyer has a utility function $v_i (x_i)$. 

We assume the utilities are all homogeneous of degree one (i.e. that $\alpha v_i (x_i) = v_i (\alpha x_i)$ for $\alpha>0$), concave, and continuous. We also assume that there exists an allocation $x$ such that $v_i (x_i) > 0$ for all buyers $i$.

This class of utilities includes the constant elasticity of substitution (CES) family of utilities, and linear utilities are special case of CES utilities where each buyer has a value $v_{ij}$ for each item and $v_i (x_i) = \sum_{j} v_{ij} x_{ij}$. Linear utilities are precisely what is assumed in many studies of online advertising markets~\citep{conitzer2019pacing,balseiro2021budget}.

Each item will be assigned a price $p_j > 0$, with $p$ being the full price vector, and each buyer has a budget $B_i > 0$. We study the quasi-linear (QL) case where leftover money is kept by the buyer; this is most natural for several real world markets such as advertising markets. Most of our results work the same for the non-QL case.

Let $\delta_i$ be the leftover budget under prices $p$ with allocation $x_i$ ($\delta_i = B_i - p^T x_i)$. The quasi-linear utility that a buyer experiences under prices $p$ and allocation $x$ is then
 $$u_i (x_i, \delta_i) = v_i (x_i) + \delta_i$$

Given a price vector $p$, the \textbf{demand} of a buyer is $$D_i(p) = \text{argmax}_{x_i: x_i^\top p \leq B_i} u_i(x_i, \delta_i).$$

A \textbf{market equilibrium} is an allocation $x^*$ and a price vector $p^*$ such that
% \begin{enumerate}
    1) allocations are demands: for each $i$ $x^*_i \in D_i(p^*)$, and
    2) markets clear: for each $j$, $\sum_i x_{ij} \leq 1$, with equality if $p_j > 0$.
% \end{enumerate}

In Fisher markets with our family of utilities, market equilibria always exist. The general problem of market equilibrium is computationally difficult~\citep{chen2009spending,vazirani2011market,chen2021complexity} but in the family of utilities we are considering the equilibrium allocation can be found via solving a version of the Eisenberg-Gale convex program (EG) \cite{eisenberg1959consensus}:
\begin{equation}
\begin{aligned}
\max_{x\geq 0,\delta \geq 0} \quad & \sum_i B_i \text{log} (v_i(x_i) + \delta_i) - \delta_i\\
\textrm{s.t.} \quad &\text{ } \sum_{i} x_{ij} \leq 1, \forall j=1,\ldots,m
\end{aligned}
\label{eq:mnw}
\end{equation}

Note that the market equilibrium definition says nothing about maximizing the sum of log utilities - the market finds equilibrium by finding prices and allocations to make supply meet demand. Rather, it is a deep result that the $x$ which solves the EG problem is also an equilibrium of the Fisher market when we take prices Lagrange multipliers on the supply constraints at the optimum~\citep{eisenberg1959consensus}.

There sometimes may be multiple equilibrium allocations in a Fisher market, however in all equilibria the prices are the same, and, importantly, the utility realized by each buyer is the same. 

This fact that Fisher equilibria are allocations that maximize the budget weighted product of utilities is an important reason they are studied in the fair division community. In the equal budget case Nash social welfare has been called an ``unreasonably effective fairness criterion'' \citep{caragiannis2016unreasonable} and is used in practice in allocation mechanisms such as Spliddit.com \citep{goldman2014spliddit}. 

\subsection{Desirable Properties of Market Equilibria}
In addition to making supply meet demand, equilibria have other desirable properties when used as allocation mechanisms. We discuss them here as later we will want to ask whether our interventions preserve them or not. We list properties of interest informally before describing them in more detail:

%\begin{enumerate}
%    \item Fisher market equilibrium outcomes are Pareto efficient
%    \item Fisher market equilibrium outcomes are envy-free
%    \item Fisher markets do not induce strategic behavior in centralized markets when markets are large
%\end{enumerate}

\textbf{Equilibria are Pareto efficient.} 
In the QL case this is defined using the leftover budget as well as the seller who receives the payments. Let $(x_i, \delta_i)$ be the bundle of goods received by buyer $i$ and their leftover budget with $(x, \delta)$ being all buyers' allocations/budgets.  

An allocation $(x, \delta)$ is \textbf{Pareto optimal} if for every alternative allocation $(x',\delta')$ such that some buyer strictly improves their utility, it must be the case that for some other buyer $i$, we have $v_i(x_i')+\delta_i' < v_i(x_i)+\delta_i$ or $\sum_i \delta_i' > \sum_i\delta_i$, meaning that the seller is worse off. Similarly, if $\sum_i \delta_i' < \sum_i \delta_i$, meaning the seller strictly improves, then it must be the case that for some buyer $i$, we have $v_i(x_i')+\delta_i' < v_i(x_i)+\delta_i$. 

\textbf{Equilibria are envy-free} when accounting for budgets. Let $(x, \delta)$ be an equilibrium allocation. For any two buyers $i, i'$ with budget ratio $\gamma = \frac{B_i}{B_{i'}}$ we always have that $v_i(x_i) + \delta_i \geq v_i (\gamma x_{i'}) + \gamma \delta_{i'}$.

\textbf{Incentives for strategic behavior are minimized when markets are large.} Consider the \emph{ Fisher mechanism}: each individual reports their valuation function to a center. The center computes the market equilibrium, and gives each individual their corresponding allocation. This is strongly related to CEEI \citep{varian1974efficiency,budish2012multi} though here utilities are quasi-linear. 

Auction-based ad markets act somewhat like a centralized Fisher mechanism. Buyers report their valuations for various events (clicks, conversions, impressions, etc...). Auctions happen for each impression and the ad system bids for each buyer using logic to adjust bids over time to spend budgets. For appropriate choices of pacing rules and auctions the resulting outcome is a Fisher market equilibrium \citep{conitzer2019pacing}. In this case, whether the individuals have incentives to lie to the center is extremely important.

At a high level, a mechanism is said to be \textbf{strategy-proof in the large} (SPL \citep{azevedo2018strategy}) if, when the market is large enough, there is little gain from lying to the center about one's valuation. 

Formally, a large market is defined as: let there be $n$ buyers, each with budget 1. The size of the market is determined by this $n$. Let there be $m$ items with generic item $j$, and let the supply of each item be $s_j n$ where $s_j$ is some constant. Then SPL is defined as for any $\epsilon > 0$ for all possible valuations $v_i$, there exists $\bar{n}$ such that if $n > \bar{n}$ the gain to any buyer of type $v_i$ from misreporting any $v'_i$ instead of their true $v_i$ is less than $\epsilon$. The standard Fisher mechanism is SPL (see the arguments in \citep{azevedo2018strategy} and the Appendix for more details).

\section{Changing Market Outcomes}
\label{sec:changing market outcomes}
We now move to our main problem. We consider a market designer that wants market allocations to satisfy certain constraints. Though in some cases (e.g. fully centralized) the designer can control quantities allocated directly. However, in many cases of interest the market designer must work through the price system or they may wish to do so due to efficiency properties.

We consider a market designer that is able to subsidize and tax purchases. Let $\bar{p}$ be an $n \times m$ matrix of \textit{price interventions}. We assume that each buyer $i$ faces price $p_j + \bar{p}_{ij}$ for item $j$ with $\bar{p}_{ij} > 0$ being a tax and $\bar{p}_{ij} < 0$ being a subsidy. Let $\bar p_i \in \R^m$ be the vector of price interventions for buyer $i$.

The demand of a buyer $i$ with base prices $p$ and interventions $\bar{p}$ is given by $D_i(p + \bar p_i)$. This lets us extend the definition of a market equilibrium. We say that given a $\bar{p}$ the triple $(x, p, \bar{p})$ is a \textbf{tax-subsidy equilibrium} if
\begin{enumerate}
    \item Each $x_i \in D_i (p + \bar{p}_i)$
    \item For all $j$ $\sum_{i} x_{ij} \leq 1$, with equality if $p_j + \bar{p}_{ij} > 0$.
\end{enumerate}

We now show that with appropriate choice of $\bar{p}$ the market designer can force the market equilibrium allocation to satisfy a set of linear constraints. The proof of the result is constructive and gives us the algorithm for computing such a $\bar{p}$ from knowledge of the market structure and the desired constraints.

First, let us formally define linear constraints. We define a linear operator $A_1 : \mathbb{R}^{n \times m} \rightarrow \mathbb \mathbb{R}^{K_1}$ which inputs the allocation matrix and yields a vector $b \in \mathbb{R}^{K_1}$. We let $A_1 x \leq b_1$ be the inequality constraints we wish to hold (the choice of less than in the inequality is arbitrary as the sign of $A$ can always be changed). We define another linear operator $A_2: \mathbb R^{n \times m} \rightarrow \mathbb R^{K_2}$ to represent equality constraints $A_2 x = b_2$. 

This gives us the constrained EG program:
\begin{align}
\begin{array}{rlll}
    \displaystyle\max_{x\geq 0,\delta \geq 0} &  \multicolumn{2}{l}{\displaystyle\sum_{i\in T} B_i \log \big(v_i(x_i) +\delta_i \big) - \delta_i} \\
    \text{s.t.} & \displaystyle \sum_{i} x_{ij} \le 1, & \forall j, \\
    & A_1 x \leq b_1, & \\
    & A_2 x = b_2. &
\end{array}    
\label{eq:con_eg}
\end{align}

We now show how the solution to the EG program can be used to construct the price interventions we seek.

\begin{prop}\label{prop:coneg_works}
Let $x^*$ be the solution to the constrained EG program \ref{eq:con_eg}. Let $\lambda^*_1, \lambda^*_2$ be the vectors of Lagrange multipliers on the inequality and equality constraints at the optimum. Define $$\bar{p}_{ij}^* = \sum_{k=1}^{K_1} A_{1kij} \lambda_{1k}^* + \sum_{k=1}^{K_2} A_{2kij} \lambda_{2k}^*.$$ Then there exists $p$ such that $(x^*, p, \bar{p}^*)$ is a tax-subsidy equilibrium where the allocation satisfies the constraints given by $A_1 x^* \leq b, A_2 x^* = b.$
\end{prop}

The full proof is relegated to the Appendix. We give the intuition here. The standard argument for using the EG program to calculate an equilibrium is that we can take the solution $x^*$ and the vector of Lagrange multipliers on the supply constraints at the optimum $\lambda^S$ and we show that $(x^*, \lambda^S)$ form an equilibrium. We extend this argument to show that the Lagrange multipliers on the other constraints work a similar way, although to recover their price interpretation they need to be `passed through' the constraint mappings $A_1, A_2$.

%We also get a straightforward corollary about the power of taxes/subsidies in a Fisher market.

%\begin{corr}
%For  any feasible allocation $\bar{x}$, we can use a convex program to compute a price intervention $\bar{p}$ such that there is a tax-subsidy equilibrium $(\bar{x}, p, \bar{p})$.
%\end{corr}

%The corollary above follows from the fact that $\lbrace x_{ij} = \bar{x}_{ij} \rbrace$ is a family of linear constraints. Of course, for ``extreme'' allocations the market designer may be required to make large taxes or subsidies.

Arbitrary systems of constraints may require a complex system of taxes/subsidies (e.g. specific taxes/subsidies for each buyer/item pair). However, the ``simpler'' the constraint, the simpler the tax/subsidy required to attain it. For example, consider the \emph{quota} constraint that we will study further in later sections. Here $B$ is a subset of buyers and $R$ is a subset of items and our goal is that, in aggregate, buyers in $B$ have at least some fixed level of exposure to $R$. This is written as the constraint $\sum_{i \in \mathcal{C}} \sum_{j \in B} x_{ij} \geq L.$

A tax/subsidy to make our equilibrium allocation satisfy this goal will only require one common subsidy $\bar{p}$ that all buyers in $B$ will face for all items in $R$. 

The full effects of taxes and subsidies in Fisher markets are not always straightforward. Typically, taxes are viewed as ways to decrease demand (and increase price) and subsidies for the opposite. However, this is more complex in markets because ``base'' prices $p$ are determined in equilibrium \textbf{given} the market designer chosen tax/subsidy schedule $\bar{p}.$

If we take any price intervention matrix $\bar{p}$, which may have taxes, and then add a constant subsidy $c$ so that the intervention becomes $\bar{p}-c$, then base prices $p$ will adjust to $p+c$ and the equilibrium will not change in any meaningful way.

Note that while there are many equivalent (in outcome) interventions, they may vary in complexity of implementation. For example, consider an intervention of the form $\bar{p}_{ij} = t$ for some $(i,j)$ pair and $0$ otherwise. It requires the market designer to just enforce a tax on a single transaction. On the other hand, an an an an an an an an equivalent intervention of the form $\bar{p}_{ij} = t+c$ for $(i,j)$ and $\bar{p}_{i'j'} = c$ everywhere else, requires taxes on \textit{every} transaction. 

Thus, while we continue to refer to price interventions as taxes and subsidies it is important to remember that it is relative, not absolute, levels of interventions that matter for outcomes.

\section{Online Price Intervention Computation}
The discussion above computed price interventions directly but required buyer valuations to be known to the market designer, as well as the ability to solve a potentially large convex program. 

We now discuss an online setting and propose a minimal-information decentralized algorithm that requires the market designer only to observe the equilibrium allocations and perform simple update procedures on their price interventions. 

We consider the online case where time is indexed by $t \in \lbrace, 0, 1, 2, \dots \rbrace.$ The set of buyers and items remains constant in each time period $t$. The set of valuations, budgets, and supplies stays constant but unknown to the market designer. At the beginning of each period, the market designer sets a tax/subsidy schedule $\bar{p}^t.$ Given these taxes/subsidies the buyers/sellers transact in some decentralized way and an equilibrium $(x^{*t}, p^{*t}, \bar{p}^t)$ arises. The equilibrium allocations and prices are observed to the market designer who can then update taxes/subsidies for the next period $t+1.$

We ask: is there a process by which the market designer can update taxes/subsidies dynamically such that they converge to taxes/subsidies that implement a given set of constraints?

We show that the answer is yes. This is a result of the fact that taxes/subsidies are Lagrange multipliers on the constraints, and thus we can take subgradient steps on them simply by observing the constraint violation. The full algorithm for online price intervention computation (OPIC) is below:

\begin{algorithm}
\caption{Online Price Intervention Computation (OPIC)}
\begin{algorithmic}
\STATE Input: Constraints $A_1, A_2$, learning rate sequence $\gamma^t$
\STATE Initialize $\lambda^0_1 \in \mathbb{R}^{K_1}_+$, $\lambda^0_2 \in \mathbb{R}^{K_2}$ 
\FOR{$t \in \lbrace 0, 1, 2, \dots \rbrace$}
\STATE Set $\bar{p}^t_{ij} = \sum_{k=1}^{K_1} A_{1kij} \lambda^t_{1k} + \sum_{k=1}^{K_2} A_{2kij} \lambda^t_{2k}$
\STATE Observe the resulting equilibrium allocation $x^*_{t} (\bar{p}^t)$
\STATE $\nabla \lambda_1  = (A_1 x^*_t - b_1)$
\STATE $\nabla  \lambda_2  = (A_2 x^*_t - b_2)$
\STATE $\lambda^{t+1}_1 = \text{max} \lbrace 0, \lambda^{t}_1 + \gamma  \nabla \lambda_1 \rbrace$
\STATE $\lambda^{t+1}_2 =\lambda^{t}_2 +  \gamma  \nabla \lambda_2$
%\STATE $\lambda^{t+1}_i = \text{min} \lbrace \lambda_{MAX}, \lambda_{i} \rbrace $
\ENDFOR
\end{algorithmic}
\end{algorithm}

Suppose the constraint set does not include any redundant constraints (if it does, eliminate constraints until it no longer does). Let $\bar{P}^*$ be the set of optimal Lagrange multipliers on the constraints, derived from the convex program \ref{eq:mnw with constraints} for the given constraint set. Under a standard assumption on learning rates from optimization theory \citep{bertsekas2015convex}, we get the following result

\begin{prop}\label{prop:opic_works}
If $\sum_{t=0}^{\infty} \gamma^t = \infty$ and $\sum_{t=0}^{\infty} (\gamma^t)^2 < \infty$ then $\lim_{t \to \infty} \bar{p}^t \to \bar{P}^*.$ 
\end{prop}

In other words: a dynamic similar to Walrasian t\^atonnement \citep{walras} on the taxes and subsidies eventually implements the constraints. The full proof is in the Appendix. There are two important steps. The first is a lemma that may be of independent interest.

\begin{lemma}\label{lem:tax_exists}
For any $\bar{p}$ there exists $(x, p)$ such that $(x, p, \bar{p})$ is a tax-subsidy equilibrium.
\end{lemma}

This lemma is a fairly simple extension of the EG convex program. However, it is important since it tells us that an equilibrium can actually be attained in each period.

We then use the results from the last section which allow us convert Lagrange multipliers at the optimum into a price intervention system. OPIC builds on this idea by considering the Lagrangified version of the constrained EG program~\ref{eq:mnw with constraints}, using the subgradient method on the resulting dual problem, and converting the Lagrange multipliers to new prices at each step. 

Using our dual formulation,  there are many other possible choices of optimization algorithms that could potentially be applied,  which may yield faster convergence than the subgradient method, and those will have their own price update rule. This is an interesting area for future research but is beyond the scope of this paper.

\section{Market Fairness Constraints}\label{sec:constraints}
We now move from a general constraint setup to focus on specific constraints that have been brought up in the fairness literature in markets and other contexts.

Many of our motivating examples are from the literature on online ad markets. In this literature, linear utility is commonly used. Buyers have vectors of valuations for each item and the total utility of a bundle is the sum of the item utilities $v_i (x_i) = v_i^T x_i$. We use this for our examples going forward.

For each of these constraints we move from asking whether a market designer \textit{can} implement it to whether they \textit{should}. In particular, we study whether good properties of market allocations (Pareto optimality, envy freeness, SPL) are preserved, we ask who wins and who loses from these interventions, and, most importantly, we ask whether these interventions achieve various motivating goals.

\subsection{Per Buyer Parity}
\citet{ali2019discrimination} consider the issue of preferential exposure of certain groups to certain kinds of job ads. A similar issue occurs in advertising for certain legally protected ad categories (e.g. housing). 

This leads to our first family, \textbf{per buyer parity} (PBP) constraints. 

Let items have a binary attribute, either $A$ or $B$. There is a subset of buyers called \textbf{constrained buyers} the set of which we denote by $C$. In order to combat disparate outcomes, constrained buyers are required to have `balanced' exposure to items. For example: in an online ad market buyers could be housing or job ads, items are ad slots for individual users, the binary attribute is some protected class on which ad delivery may not differ.

This leads to the following linear constraints: for any $i \in C$ we require that $$\sum_{j \in A} x_{ij} = \alpha \sum_{j \in B} x_{ij}.$$ 

Other examples in the literature can also be expressed as choices of $\alpha$. For example, the equal exposure constraint sets $\alpha = \frac{N_{B}}{N_{A}}$ where $N_K$ is the size of each group. The choice of $\alpha$ strongly depends on context, for example \citet{ali2019discrimination} study job ads and suggest using a weight of number of expected qualified candidates in each group rather than simply group size.

For the purpose of the discussion here we set the total number of $A$ and $B$ items equal and let $\alpha=1$. This simplifies notation and lets us state our main results and counter-examples. 

\begin{table*}%
\begin{tiny}
  \centering
  \subfloat{ \begin{tabular}{lcr}
    \toprule
    V &   &   \\
    \midrule
    Buyer & Item $A$ & Item $B$ \\
    C & 1.5 & .4 \\ 
    C & .4 & 1.5 \\ 
    U & 5 & 2 \\
    U & 2 & 5 \\
    \bottomrule \\
      &   &     
  \end{tabular}}%
  \qquad
  \subfloat{ \begin{tabular}{lcrrr}
    \toprule
    $X^{EG}$ &   &  &  & \\
    \midrule
    Buyer & Item $A$ & Item $B$ & $\delta_i$ &$u_i$ \\
    C & .33 & 0 & .5 & 1 \\ 
    C & 0 & .33 & .5 & 1\\ 
    U & .66 & 0 & 0 & 3.33\\
    U & 0 & .66 & 0 & 3.33\\
    \bottomrule \\
    Price & 1.5 & 1.5 & -
  \end{tabular}}
  \subfloat{ \begin{tabular}{lcrrr}
    \toprule
    $X^{PBFP}$ &   &  &  & \\
    \midrule
    Buyer & Item $A$ & Item $B$ & $\delta_i$ & $u_i$ \\
    C & 0 & 0 & 1 & 1 \\ 
    C & 0 & 0 & 1 & 1\\ 
    U & 1 & 0 & 0 & 5\\
    U & 0 & 1 & 0 & 5\\
    \bottomrule \\
    Price & 1 & 1 & -
  \end{tabular}}
  \caption{An example of a set of valuations with budgets set to $1$ where imposing PBFP constraints leads to both parity constrained (C) buyers exiting the market completely and all items going to the unconstrained (U) buyers. Under the constraints, the $U$ buyers get both items and get them cheaper than in the original equilibrium, so they are better off.\label{exit_market} }%
  \end{tiny}
\end{table*}

The implicit or explicit goal of PBP constraints is generally to help some underlying disadvantaged group. Therefore, any notion of Pareto optimality must also include this group.

Formally, given an allocation $X$. For a buyer $i$ let the \textbf{$i$-disadvantaged item group} be the one that has less exposure. Let the \textbf{aggregate disadvantaged item group} be the one that has less exposure to $C$ buyers. Typically, parity constraints are motivated by the goal of `improving' the outcomes of this aggregate disadvantaged group. For this reason we will refer to the aggregate disadvantaged group as the \textbf{protected} item group.

Let $P$ be the protected group of items. Let $\sum_{i \in C} \sum_{j \in P} x_{ij}$ be the \textbf{aggregate exposure of the protected group to the $C$ buyers}. An allocation $(x, \delta)$ is \textbf{buyer-protected-item Pareto optimal} if for every alternative allocation $(x',\delta')$ such that some buyer strictly improves their utility or such that the protected group $P$ gains more aggregate exposure to $C$ buyers, it must be the case that for some other buyer $i$, we have $v_i(x_i')+\delta_i' < v_i(x_i)+\delta_i$ or $\sum_i \delta_i' > \sum_i\delta_i$, meaning that the seller is worse off. Similarly, if $\sum_i \delta_i' < \sum_i \delta_i$, meaning the seller strictly improves, then it must be the case that for some buyer $i$, we have $v_i(x_i')+\delta_i' < v_i(x_i)+\delta_i$ or the protected group loses exposure to $C$ buyers. Finally, if $P$ gets more aggregate exposure to $C$ buyers it must be that either some buyer or the seller is strictly worse off. We refer to the original definition in Section \ref{fisher_markets} as \textbf{buyer-only Pareto optimality}. 

We now ask whether per-buyer parity constraints keep good market properties. 

\begin{prop}\label{pbfp_pareto}
PBP constrained equilibria are neither buyer-only Pareto optimal nor buyer-protected-item Pareto optimal.
\end{prop}

We show counterexamples for this in shown the Appendix Example Table \ref{example_pareto_pbfp}. While lack of buyer-only Pareto optimality may not be surprising (since often the goal of the intervention is to increase exposure of the protected group), the lack of buyer-protected-item Pareto optimality is more troubling.

In addition, imposing parity constraints on a single buyer also imposes pecuniary externalities (i.e. second order effects) on other buyers. The net of these second order effects may be large. This means, winners and losers are not clear.

Consider two buyers $i, k$. Suppose we add parity constraints only to $i$, then $i$ is made worse off. However, the addition of parity constraints changes $i's$ demand, it reduces demand for $i's$ originally advantaged group and increases it for $i's$ originally disadvantaged group. This means the second order effect can be positive or negative from buyer $k$'s perspective. In addition, the constraint may completely miss the mark of its original goal.

\begin{prop}\label{pbfp_decrease}
Adding PBP constraints can decrease the exposure of the protected item group to $C$ buyers relative to the original equilibrium.
\end{prop}

We see this in Example Table \ref{exit_market}. Adding parity constraints leads constrained buyers to exit the market. This means after the imposition of PBP constraints \textit{all items receive less exposure} to the constrained buyers. However, parity constraints also decrease competition and thus prices paid by unconstrained buyers. 

In the context of job advertising the interpretation of Example \ref{exit_market} is that adding parity requirements to job ads leads to a particular type of parity - nobody receives any job ads at all because job advertisers find it more profitable to keep their budget than participate in the market. The main improvement comes to `regular' advertisers who now face less competition and thus lower prices.

Note that even which group of buyers wins or loses is not guaranteed. For example, all buyers are worse off in the example in Appendix Table \ref{example_everyoneworse_pbfp}.

Shifting focus to envy-free and SPL properties, it is clear that in general there may be envy between a constrained and an unconstrained buyer who share the same valuation function (since the constrained buyer effectively has a strictly smaller choice set). However, within group this is not the case:

\begin{prop}\label{pbfp_envy}
If buyers $i, i'$ are both constrained (or both unconstrained) then $i$ does not have (budget adjusted) envy for $i'$.
\end{prop}

Finally, and strongly related to the envy-free property, we see that SPL continues to hold: 

\begin{prop}\label{pbfp_spl}
The Fisher market with PBFP constraints and constant fraction $\gamma$ of constrained buyers who cannot misreport whether they are constrained is SPL.
\end{prop}

We relegate proofs of both propositions to the Appendix. 

\subsection{Per Item Parity}
We now turn to the `transpose' of the PBP constraint. \citet{geyik2019fairness} consider the case of job candidate search and the goal of creating a balanced slate of candidates in any search query. 

The market version of such a constraint is that items have a balanced exposure to different types of buyers. For example, this can be used to balance the exposure of content consumers on certain types of content producers \citep{singh2018fairness,geyik2019fairness} or making sure that the impressions of ads of a certain buyer type are evenly distributed across groups of individuals.

We refer to such constraints as per item parity constraints (PIP). 

Formally, buyers are one of two types, either $A$ or $B$. There is a subset of items which are constrained which we denote by $C$. Formally the constraint is given by $$\forall j \in C, \text { } \sum_{i \in A} x_{ij} = \alpha \sum_{i \in B} x_{ij}.$$ As above this can always be expanded to different item groups $C_1, C_2, \dots$ with their own values of $\alpha$. 

As before, there are many choices of $\alpha$. For simplicity, we again focus on parity of exposure, i.e. $\alpha = \frac{N_A}{N_B}$. Note though, that various other choices of $\alpha$ in the buyer parity case often have natural equivalents in the item parity case. 

Given an allocation $X$, we call the buyer group $G$ which has less exposure to items $C$ the \textbf{disadvantaged buyer group}. With this in mind, we now focus on the welfare consequences of implementing per item fractional parity constraints. 

\begin{prop}
PIP constraints do not lead to buyer-only Pareto optimal allocations. Originally disadvantaged buyers may be worse off after constraints are implemented.
\end{prop}

To put this result into context, consider the use of an ad market or market-based mechanism like CEEI for recruiters. Here each item is a recruiter impression and each buyer is a job applicant. This means we take a `job applicant'-centric perspective, i.e., we focus on maximizing the applicant's goals with a Nash social welfare criterion as a distributional goal. On the other hand, to avoid discrimination on the recruiter side, we enforce parity of exposure to each recruiter exposure across some binary attribute of the applicants. Then, it is possible that the job applicants we are trying to help are worse off after this `anti-discrimination' intervention.

It is easy to construct an example where all buyers are worse off under PIP by simply considering that PIP can force an equal split of items which clearly can be Pareto dominated by many allocations (see Table \ref{example_everyoneworse_pifp} in the Appendix).

\begin{table*}%
\begin{tiny}
  \centering
  \subfloat{ \begin{tabular}{lcr}
    \toprule
    V &   &   \\
    \midrule
    Buyer & Item $C$ & Item $U$ \\
    $A_1$ & 2 & 1 \\ 
    $A_2$ & 2 & 1.5 \\ 
    $B_1$ & 3 & 2 \\
    $B_2$ & 3 & 2 \\
    \bottomrule 
  \end{tabular}}%
  \qquad
  \subfloat{ \begin{tabular}{lcrrr}
    \toprule
    $X^{EG}$ &   &  &  & \\
    \midrule
    Buyer & Item $C$ & Item $U$ & $\delta_i$ & $u_i$ \\
    $A_1$ & .167 & 0 & .66 & 1 \\ 
    $A_2$ & 0 & .75 & 0 & 1.12 \\ 
    $B_1$ & .417 & .125 & 0 & 1.5\\
    $B_2$ & .417 & .125 & 0 & 1.5\\
    \bottomrule 
  \end{tabular}}
  \subfloat{ \begin{tabular}{lcrrr}
    \toprule
    $X^{PIP}$ \text{ or} &  $X^{AEF}$   &  &  & \\
    \midrule
    Buyer & Item $C$ & Item $U$ & $\delta_i$ & $u_i$ \\
    $A_1$ & .382 & 0 & .24 & 1 \\ 
    $A_2$ & .118 & .42 & .14 & 1\\ 
    $B_1$ & .25 & .29 & 0 & 1.33 \\
    $B_2$ & .25 & .29 & 0 & 1.33 \\
    \bottomrule 
  \end{tabular}}
  \caption{An example of a set of valuations with budgets set to $1$ where all buyers prefer $C$ to $U$ and $A$ buyers are disadvantaged originally. However, adding PIP constraints fails to improve their utility and decreases the utility of $A_2$. Note that this is equivalent an AEF constraint of requiring $.5$ exposure of $C$ to $A$ buyers.}%
  \label{example_crowd_out_pifp}%
  \end{tiny}
\end{table*}

However, PIP can backfire in a different way as we see in Example Table \ref{example_crowd_out_pifp}. Here \textit{all} buyers prefer item $C$ to item $U$, and $A$ buyers are less exposed to it than $B$ buyers in equilibrium. Implementing PIP fails to help $A$ buyers. While it increases exposure of $A_1$ and $A_2$ buyers to $C$ it strictly decreases the welfare of buyer $A_2$ because of increased competition from $B$ buyers for item $U$.

As with the PBFP constraints we see that envy free and SPL properties continue to hold.

\begin{prop}
If buyers $i, i'$ are both the same binary label then $i$ does not have (budget adjusted) envy for $i'$ in PIP constrained equilibrium.
\end{prop}

\begin{prop}
The Fisher market with PIP constraints and a constant fraction $\gamma$ of buyers as $A$ buyers where no buyers can misreport their group affiliation is SPL.
\end{prop}

The full proofs of these propositions are in the Appendix.

\subsection{Aggregate Exposure Floor Constraints}
We now turn to studying an extremely simple constraint: we set a floor below which aggregate exposure of some buyers to some items cannot fall. We let $C$ be a subset of constrained buyers and $P$ a subset of protected items. 

The \textbf{aggregate exposure floor constraint} (AEF) is written as $$\sum_{i \in \mathcal{C}} \sum_{j \in P} x_{ij} \geq L$$ where $L$ is our floor level. We require that $L$ is jointly feasible with the supply constraints. 

We see that here at least some form of Pareto optimality is preserved:

\begin{prop}\label{AEFPO}
Adding AEF constraints to EG guarantees buyer-protected-item Pareto optimality. However, it does not guarantee buyer-only Pareto optimality. 
\end{prop}

The proof and counter-example for this proposition is can be found in the Appendix. 

The AEF constraint can be used as an alternative to both PBP and PIP constraints above. The goal of the floor is to increase exposure of buyers $C$ to items $P$. In the context of ads this can be motivated by improving welfare for $P$ (showing more of some ad type to some users increases user utility) \textbf{or} by improving welfare for $C$ (subsidizing some ad type bids for some users is to help those advertisers).

Proposition \ref{AEFPO} shows that it is more natural to wield the AEF constraint when the market designer's goal is some form of `supply-side improvement' than when it is some `buyer-side improvement.' 

An example of a case where AEF may be used for `buyer-side' improvements is the job applicant setting the designer may want to increase exposure of some applicant group (buyers) to some group of recruiters (items). 

The AEF constraint is implemented via pricing by a blanket subsidy of $\bar{p}$ for all buyers in $C$ for all items in $P$. It seems intuitive that subsidizing buyers can only improve their utility, but this is actually incorrect:

\begin{prop}
Adding AEF constraints can decrease the utility of buyers in $C$.
\end{prop}

This follows from the fact that AEF constraints set with the right threshold can simulate PIP constraints. The PIP constrained equilibrium in Table \ref{example_crowd_out_pifp} is also achievable with AEF constraints. Again, adding AEF constraints helps buyer $C_1$ get more of item $P$ but buyer $C_2$ prefers the non-protected item and, after implementation of AEF constraints, faces increased competition from buyers outside of the $C$ set.

Thus, using AEF constraints purely as a way to improve buyer outcomes can sometimes backfire.

As with other PIP, PBP we find similar envy/SPL results:

\begin{prop}
The Fisher market with AEF constraints is SPL when no buyers can misreport their group affiliation is per-group envy free and SPL.
\end{prop}

Again, we relegate the proofs to the Appendix. Note that our SPL results require that the size of the constrained group grows to infinity as the market gets large - this means that the subsidy is done at the buyer/item group level and is, in the infinite limit, not affected by the report of an individual buyer. We cannot in general have SPL results for individualized floor constraint per buyer since this would imply individualized subsidies and would incentivize buyers to under-report their valuation in order to increase their subsidy.

\subsection{Summary of Fairness Constraints}
In Table \ref{table:summary_con} we summarize the discussion of each constraint. In addition, we provide the form that the price intervention from Section \ref{sec:changing market outcomes} takes for that constraint. From this it is easy to see the OPIC update rule.

%\begin{center}
%\begin{table}\label{table:summary_con}
%\begin{small}
%
%\begin{tabular}{||p{3cm} | p{3.5cm} | p{2.75cm}  | p{3.5cm}||} 
%\hline
% Constraint & Price Intervention Form & Guarantees & Possible Issues \\ [0.5ex] 
% \hline\hline
% \makecell{Buyer Parity \\{\small (Protected Item Group) }}& \makecell{For each buyer $i \in C$, \\tax $\bar{p}_i$ on $i$-advantaged items,\\ same size subsidy $\bar{p}_i$ \\on $i$-disadvantaged items} & \makecell{Within buyer-group\\ envy-free, SPL} & \makecell{May decrease exposure of \\protected group relative\\ to original equilibrium} \\ 
% \hline
% \makecell{Item Parity \\{\small (Protected Buyer Group)} } & \makecell{For each item $j \in C$ and buyers\\ $i \in A$, $k \in B$\\ we have $\bar{p}_{ij} = \bar{t} = -\bar{p}_{kj}$\\ for some fixed $\bar{t}$} & \makecell{Within buyer-group \\envy-free, SPL} & \makecell{May decrease the \\utility of originally\\ disadvantaged buyers} \\
% \hline
% \makecell{Aggregate Floor\\ {\small (Can be motivated by} \\{\small protected item} \\{\small \textit{or} protected buyer)}} & \makecell{A fixed subsidy $\bar{p}_{ij} = \bar{s} \geq 0$\\ for all buyers $i \in C$, items $j \in P$} & \makecell{Buyer-item\\ Pareto-optimal, \\within buyer-group\\ envy-free, SPL} & \makecell{Can decrease utility\\ of buyers in $C$ }\\
% \hline
%\end{tabular}
%\end{small}
%\caption{Summary of properties of constraints studied in Section \ref{sec:constraints}.}\label{table:summary_con}
%\end{table}
%
%\end{center}

\begin{center}
\begin{table}\label{table:summary_con}
\begin{small}

\begin{tabular}{||p{3cm} | p{3.5cm} | p{2.75cm}  | p{3.5cm}||} 
\hline
 Constraint & Price Intervention Form & Guarantees & Possible Issues \\ [0.5ex] 
 \hline\hline
Buyer Parity (Protected Item Group) & 
For each buyer $i \in C$, tax $\bar{p}_i$ on $i$-advantaged items, same size subsidy $\bar{p}_i$ on $i$-disadvantaged items & 
Within buyer-group envy-free, SPL & 
May decrease exposure of protected group relative to original equilibrium \\ 
 \hline
 
Item Parity (Protected Buyer Group) & 
For each item $j \in C$ and buyer $i \in A$, $k \in B$ we have $\bar{p}_{ij} = \bar{t} = -\bar{p}_{kj}$ for some fixed $\bar{t}$ & 
Within buyer-group envy-free, SPL & 
May decrease utility of originally disadvantaged buyers \\
 \hline
 
Aggregate Floor (Can be motivated by protected item \textit{or} protected buyer) & A fixed subsidy $\bar{p}_{ij} = \bar{s}$ for all buyers $i \in C$, items $j \in P$ & Buyer-item Pareto-optimal, within buyer-group envy-free, SPL & Can decrease utility of buyers in $C$ 
\\
 \hline
\end{tabular}
\end{small}
\caption{Summary of properties of constraints studied in Section \ref{sec:constraints}.}\label{table:summary_con}
\end{table}

\end{center}

\section{Experiments}
Our results about the constraints are mostly negative counter-examples. However, it is not clear whether these counter-examples are knife-edge or whether these failures are common. We now evaluate some of our results in randomly generated markets.  

We generate random small markets with $8$ buyers split into two groups and $10$ items split into two groups. The valuations for each buyer for each item are uniform $[0,1].$ We use rejection sampling to make sure that the following conditions are met:

\begin{enumerate}
    \item In PBP experiments we require that initial (i.e. in unconstrained equilibrium) aggregate exposure of item group $A$ to $C$ buyers is at most $.7$ of the exposure of item group $B$ to $C$ buyers
    \item In PIP experiments we require that initial aggregate exposure of $C$ items to $A$ buyers is at most $.7$ of the exposure of $C$ items to $B$ buyers
    \item In AEF experiments we require that less than $15 \%$ of $A$ item supply is allocated to $C$ buyers, our constraint targets an exposure floor of $30 \%$
\end{enumerate}

We first show that OPIC converges quickly in these markets. While our theorem focuses on last iterate convergence with decreasing learning rate, we will look at what is arguably a more natural metric for real world markets. We let $\tilde{x}_T = \frac{1}{T} \sum_{t=0}^{T} x^*_{t}$ be the time average allocation up to time $T$ and consider how well this time averaged allocation satisfies the constraints.

We construct $50$ markets for each constraint and run OPIC for $50$ epochs with a constant learning rate. Figure \ref{fig:exp_results} shows how much the average allocation $\tilde{x}_T$ violates each constraint in absolute value. For each constraint we see that $\tilde{x}_T$ converges quickly to approximately satisfy the constraint even without the learning rate schedule required for last iterate converges per Proposition \ref{prop:opic_works}.

We generate random markets as above and solve for the standard equilibrium $x^*$ as well as the constrained one $x^*_c$ for various constraints. We ask several questions:

We first ask: How much can we Pareto improve the allocation $x^*_c$ for various notions of Pareto improvement? We consider only changing allocation around buyer, so we use $v_i \cdot x^*_c$ as the baseline utility level. Our first notion is the \textbf{buyer Pareto gap} \citep{kroer2019computing} which measures the gain in buyer social welfare we can achieve while ensuring that each buyer attains at least their utility in $x^*_c$. We extend this notion to the \textbf{buyer-item Pareto gap} by requiring that the aggregate exposure of $C$ buyers to disadvantaged items (in the case of PFP, AEF) or the aggregate exposure of disadvantaged buyers to $C$ items (in the case of PIP) also stays at least equal to $x^*_c$. 

We also look at utility outcomes from imposition of constraints. For each buyer $i$ we look at their utility in $x^*$ compared to $x^*_c$. To normalize notation we refer to constrained buyers in $PFP, AEF$ as well as originally disadvantaged buyers in PIP as \textbf{target buyers} and the rest as \textbf{other buyers}.

\begin{figure}[h!]
\includegraphics[scale=.6]{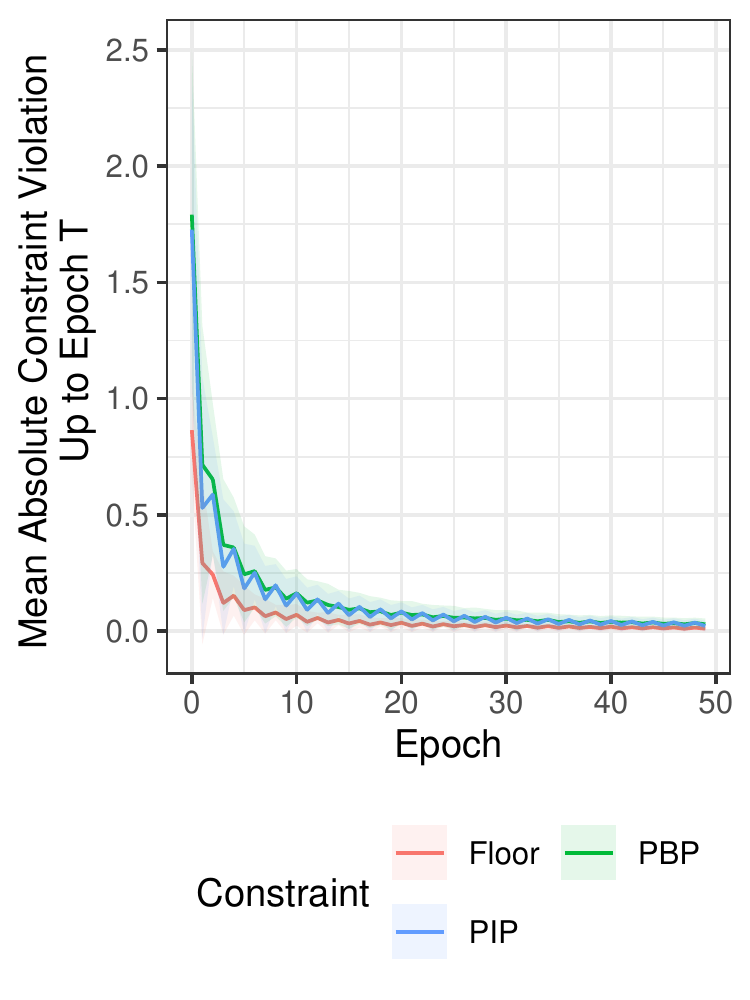} 
\includegraphics[scale=.6]{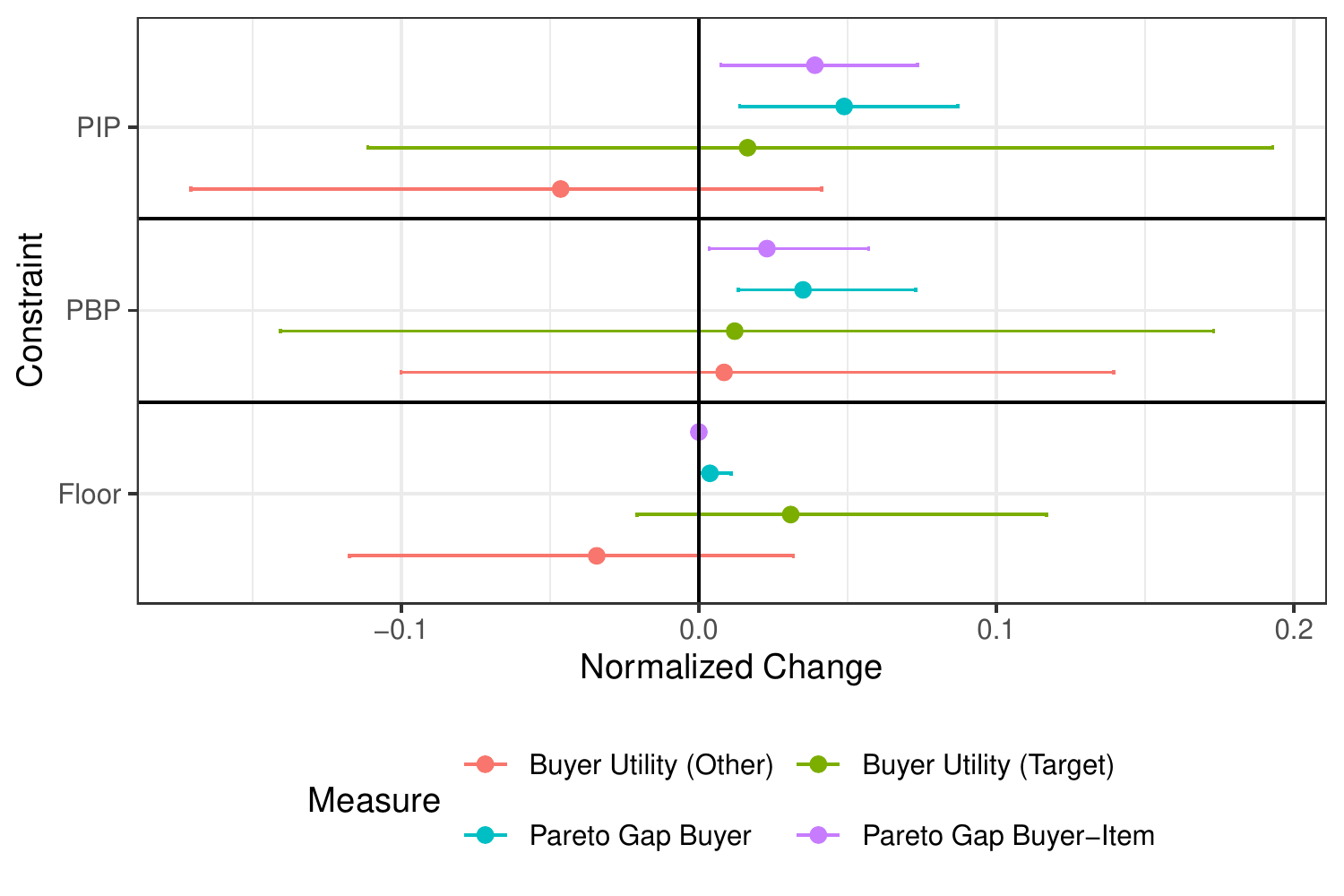}
\caption{Left panel shows that OPIC with a fixed learning rate $(.2)$ converges quickly in our random markets. Right panel shows the effects of adding constraints to random market equilibria. Error bars indicate $.05$ and $.95$ percentiles taken over 100 random markets per constraint.}\label{fig:exp_results}
\end{figure}

We see the $x^*$ to $x^*_c$ comparisons in the right panel of Figure \ref{fig:exp_results}. Dots reflect averages over $100$ markets, error bars reflect $5^{th}$ and $95^{th}$ percentiles. In the case of both Parity constraints we see that both target and non-target buyers can have their utility increased or decreased, whereas in the case of floor constraints a relatively small fraction of Target buyers lose utility and we mostly have a transfer of utility from Other buyers to Target ones. Floor constraints are buyer-item Pareto optimal but in the random markets they are also nearly buyer-only Pareto optimal as well. By contrast both parity constraints leave Pareto improvements on the table in both the buyer only and the buyer-item sense.

We note that in none of our randomly generated markets did we observe the decrease in aggregate exposure from example \ref{exit_market}. This is driven by the fact that in all of our valuations are uniform, identically distributed between groups so there is almost always full budget spend by all buyers. Characterizing the types of value distributions which lead to decreases in exposure is an interesting future problem but beyond the scope of this paper.

\section{Conclusion and Future Directions}
We have studied how a market designer can implement arbitrary constraints on market allocations using price interventions. We have shown how these interventions can be set dynamically when the market designer has limited information about market participants. Finally, we studied what can and cannot be guaranteed when various constraints from the literature are implemented in markets.

We do not make claim that one type of intervention is \textit{always} superior to another. The choice of intervention will, of course, be specific to the market being studied, the goals of the designer, and the sensitivity to various types of trade-offs. Rather, we hope that our work shows that in market settings there can be many second order effects which make interventions work differently from what is intended or expected and gives a framework for trying to make principled trade-off decisions.

\bibliographystyle{ACM-Reference-Format}
%\bibliography{references} 
\bibliography{template.bbl}   %%% Uncomment this line and comment out the ``thebibliography'' section below to use the external .bib file (using bibtex) .
\clearpage % to make it easier to remove the appendix for submission
\appendix

\section{Appendix}
\subsection{Consequences of fairness interventions}

As discussed in Section \ref{sec:related}, there is a growing body of work analysing how fairness constraints affect groups or welfare beyond a first-order effect of equalizing outcomes. We describe these results in more details here and discuss how our results fit in this larger landscape.

The main line of work has been done in the context of binary classification, and a first type of result is based on highlighting potential misalignment between what is considered ``fair'' and what is considered ``welfare''. Early work discussed the ``cost of fainess'' as reduced predictive performance resulting from enforcing fairness constraints \citep{corbett2017algorithmic,menon2018cost}. The underlying argument is that maximizing accuracy under \emph{any} constraints should decrease accuracy compared to an unconstrained baseline. Other authors discussed the relationship between welfare based on long-term outcomes \citep{kasy2021fairness,liu2018delayed,weber2022enforcing}, and how fairness constraints as proposed in the literature might not align with them. For instance, in lending, standard fairness constraints tend to consider being accepted as a favorable outcome, but in reality it may not be a favorable outcome on the long run if the individual defaults \citep{liu2018delayed}. In such cases, fairness constraints to increase acceptance of one group might come at the cost of increasing the default risk, which would be detrimental for that group in the long run.

The results we obtain are different from this line of work for two reasons. First,  they apply to markets rather than classification tasks: The ``backfire effects'' are second-order effects of budget/supply constraints that do not exist in classification. For instance, in the case of per-item parity constraints, enforcing parity might go against individual utility because we may force buyers to exhaust their budget on an item they like less. 
Second, our results also are of a different type. In our results, ``backfire effects'' do not ncessarily come from miaslignment between the fairness constraints and (long-term) welfare. In the case of per-buyer parity, the possible issue comes from a difference between the local effect of each constraint and the aggregate effects of all constraints. Despite the subsidies for each buyer, a possible second-order effect of parity constraints is to change the equilibirum prices by increasing the cost of purchases of constrained buyers for all item groups.

Additionally, \citet{nilforoshan2022causal} consider fairness constraints defined through long-term outcomes, which they call causal notions of fairness. They consider selection tasks (binary classification with a constraint on acceptance rate) and show that optimal decision policies satisfy fairness constraints only on a set of measure 0 of individual utility functions. 

Finally, another line of work studies the effect of fairness interventions in the presence of strategic agents. The main line of work focuses on stylized job markets, studying long-term effects of affirmative action policies \citep{fang2011theories,mouzannar2019fair,emelianov2022fair} given simple models of recruiters and applicants. While both work study some form of equilibrium, we focus on market equilibria where agents have fixed preferences, rather equilibria induced y strategic behavior in simpler settings.

\section{Proof of Results}
We organize the section by equilibrium property.

\subsection{Tax Subsidy Construction}
\begin{proof}[Proof of Proposition \ref{prop:coneg_works}]
The proof is based on a straightforward generalization of the argument for the Eisenberg-Gale convex program with the additional constraint that $x\in \mathcal D$:
\begin{align}
\begin{array}{rlll}
    \displaystyle\max_{x\geq 0,\delta \geq 0} &  \multicolumn{2}{l}{\displaystyle\sum_{i\in T} B_i \log \big(v_i(x_i) +\delta_i \big) - \delta_i} \\
    \text{s.t.} & \displaystyle \sum_{i} x_{ij} \le 1, & \forall j, \\
    & A_1 x \leq b_1, & \\
    & A_2 x = b_2. &
\end{array}    
\label{eq:mnw with constraints}
\end{align}

Recall that in the standard EG argument, the prices correspond to the Lagrange multipliers $\{p_j\}_{j\in[m]}$ on the supply constraints at the optimum.

Here, we end up with additional Lagrange multipliers $\lambda_1\in \mathbb R^{K_1}_+$ and $\lambda_2 \in \mathbb R^{K_2}$.

To construct a tax-subsidy equilibrium we construct price interventions $\bar{p}_{ij}^* = \sum_{k=1}^{K_1} A_{1kij} \lambda_{1k}^* + \sum_{k=1}^{K_2} A_{2kij} \lambda_{2k}^*$ using the Lagrange multipliers at optimality, set prices equal to the Lagrange multipliers $p_j^*$ on the supply constraints, and take the optimal solution $x^*$ of \cref{eq:mnw with constraints} as the corresponding allocation.

We extend the EG argument for why the corresponding $x^*,p^*$ from \cref{eq:mnw} constitute a market equilibrium, we use KKT conditions as well as a generalization of Euler's identity for homogeneous functions to show that $(x^*, p^*, \bar p^*)$ form a tax-subsidy equilibrium.

  Let $x,\delta$ be an optimal solution to \cref{eq:mnw with constraints}. Such a solution exists by our assumptions.
  We use $\nabla_{j} u_i(x_i,\delta_i)$ to denote the $j$'th component of an arbitrarily-selected subgradient of $v_i(x_i)$.
  Let $\nu_i$ and $\mu_{ij}$ be the Lagrange multipliers for $\delta_i \geq 0$ and $x_{ij}\geq 0$ respectively.

  We start by showing that each buyer $i$ spends their budget exactly. We let $u_i = u_i(x_i,\delta_i)$.
  Now consider the KKT conditions for \cref{eq:mnw with constraints} (we leave out primal and dual feasibility conditions here since they are straightforward):
  \begin{enumerate}
    \item (Stationarity) For each $x_{ij}$ we have
    \begin{align*}
    (\nabla_{j} u_i(x_i,\delta_i)) \frac{B_i}{u_i} - p_j - \sum_{k=1}^{K_1} A_{1kij} \lambda_{1k}& \\ - \sum_{k=1}^{K_2} A_{2kij} \lambda_{2k} + \mu_{ij} &= 0,
    \end{align*}
    and for $\delta_i$ we have $\frac{B_i}{u_i} + \nu_i = 1$.
    \item (Complementary slackness) 
    \begin{align}
      x_{ij}\mu_{ij}=0,\ & \forall i,j,\\
      \delta_i \nu_i = 0, \forall i,\\
      p_j(1 - \sum_i x_{ij}) = 0,\ & \forall j,\\
      \lambda_1^\top (b_1 - A_1x)=0,\\
      \lambda_2^\top (b_2 - A_2x)=0.
    \end{align}
  \end{enumerate}

  Rewriting stationarity, using the fact that $x_{ij}\mu_{ij}=0$, and using our definition of $\bar p_{ij}$, we get
  \[
    (\nabla_{j} u_i(x_i, \delta_i)) \frac{B_i}{u_i} \leq  p_j+\bar p_{ij},
  \]
  where equality holds if $x_{ij} > 0$. Let $\tilde p_{ij} = p_j+\bar p_{ij}$.
  Now we multiply each side of the rewritten stationarity condition by $x_{ij}$, sum over $j$, and add $\delta_i \frac{B_i}{u_i} = \delta_i$ to get
  \[
    \sum_j x_{ij} (\nabla_{j} u_i(x_i,\delta_i)) \frac{B_i}{u_i} + \delta_i \frac{B_i}{u_i}
    = \sum_j x_{ij}\tilde p_{ij} + \delta_i.
    %= \delta_i = x_{ij}\bar p_{ij} + \delta_i.
  \]
  Now we can apply the \emph{generalized Euler identity for subdifferentials}~\citep{yang2008generalized}:
  let $g \in \partial v_i(x_i)$ belong to the subdifferential of $v_i$ at $x_i$. Then the generalized Euler identity states that $\sum_{j} g_j x_{ij} = v_i(x_i)$.

  Applying this identity we get
  \[
     \sum_j x_{ij}\tilde p_{ij} + \delta_i =  u_i \frac{B_i}{u_i} = B_i.
  \]
  Since the left-hand side is expenditure, we get that buyer $i$ spends their budget exactly.

  Next, we need to show that for every alternative bundle $x_i',\delta_i'$ such that $\sum_j x_i'\tilde p_{ij} + \delta_i' \leq B_i$, we have $u_i \geq u_i(x_i', \delta_i')$.
  By concavity of $u_i$ we have 
  \begin{align*}
    u_i(x_i',\delta_i')& - u_i \leq \sum_j \nabla_j u_i(x_i,\delta_i)(x_{ij}' - x_{ij}) + \delta_i' - \delta_i \\
    &=\frac{u_i}{B_i} {\sum}_j  (\tilde p_{ij} - \mu_{ij}) (x_{ij}' - x_{ij}) + \frac{u_i}{B_i} (1-\nu_i)(\delta_i' - \delta_i) \\
    &= \frac{u_i}{B_i} \big( {\sum}_j  (\tilde p_{ij} - \mu_{ij}) x_{ij}' +  (1-\nu_i)\delta_i'  - B_i\big) \\
    &\leq \frac{u_i}{B_i} \big( {\sum}_j   \tilde p_{ij}  x_{ij}' +  \delta_i'  - B_i\big) \\
    &\leq 0
  \end{align*}
where the first equality follows by stationarity, and the second by complementary slackness and the fact that $(x_i, \delta_i)$ spends the budget exactly.
The second-to-last inequality follows because all variables are positive, and the last inequality follows by budget feasibility.
Thus we have shown that each buyer $i$ receives a bundle $x_i$ belonging to their demand set $D_i(p+\bar p_i)$.

Finally we need to check the market clearing condition: this follows immediately from complementary slackness on $p_j$ and primal feasibility.
\end{proof}

\subsection{OPIC}
\begin{proof}[Proof of Lemma \ref{lem:tax_exists}]
Consider the modified EG program 

\begin{align}
\begin{array}{rlll}
    \displaystyle\max_{x\geq 0,\delta \geq 0} &  \multicolumn{2}{l}{\displaystyle\sum_{i\in T} B_i \log \big(v_i(x_i) +\delta_i \big) - \delta_i - \sum_{ij} \bar{p}_{ij} x_{ij}} \\
    \text{s.t.} & \displaystyle \sum_{i} x_{ij} \le 1, & \forall j,
\end{array}    
\end{align}

We can apply the same argument as in the proof of Proposition \ref{prop:coneg_works} to see that this EG program produces an allocation where the supporting prices faced by each buyer $i$ are $p + \bar{p}_i$. 
\end{proof}

\begin{proof}[Proof of Proposition \ref{prop:opic_works}]
Recall that we assume that there exists a supply-feasible allocation $x$ such that $A_1 x < b_1$, $A_2x=b_2$, and $u_i(x_i) > 0$ for all buyers $i$.

With this assumption, we have that Slater's condition holds, and in turn this means that strong duality holds. Thus, we can equivalently solve \cref{eq:mnw with constraints} by solving the following saddle-point problem:
\begin{equation}
\begin{aligned}
    \min_{\lambda_1\geq 0,\lambda_2} \bigg\{\max_{x\geq 0,\delta \geq 0} &\sum_{i\in T} \left(B_i \log \big(v_i(x_i) +\delta_i \big) - \delta_i\right) 
    & - \lambda_1^\top (A_1x-b_1) - \lambda_2^\top (A_2x - b_2) \\
    \text{s.t.} &  \sum_{i} x_{ij} \le 1,  \forall j \bigg\}
\end{aligned}
\label{eq:Lagrangian dual mnw}
\end{equation}

Let $g(\lambda_1,\lambda_2)$ equal the value of the inner maximization problem for a fixed $\lambda_1,\lambda_2$. This is a convex minimization problem in $\lambda_1,\lambda_2$.

For any fixed $\lambda_1,\lambda_2$, we have from Danskin's theorem that the subgradients of $g$ are given by $(A_1x(\lambda_1,\lambda_2) - b_1, A_2x(\lambda_1,\lambda_2) - b_2)$, where $x(\lambda_1,\lambda_2)$ is any optimal solution to the inner maximization problem given $\lambda_1,\lambda_2$.

Note that our subgradients are obviously bounded, since the total violation possible is bounded. We can apply standard results from convex optimization show that if $$\lim_{t\rightarrow \infty} \gamma^t = 0$$ and $$\sum_{t=0}^\infty \gamma^t = \infty$$ then the iterate sequence $\{\lambda_1^t,\lambda_2^t\}$ converges to some optimal Lagrange multipliers $\lambda_1^*,\lambda_2^*$ in the sense that $$\lim_{t\rightarrow \infty} g(\lambda_1^t,\lambda_2^t) = g(\lambda_1^*,\lambda_2^*).$$ Moreover, if $\sum_{t=0}^\infty \gamma_t^2 < \infty$, then $\lim_{t\rightarrow \infty} (\lambda_1^t,\lambda_2^t) \rightarrow (\lambda_1^*,\lambda_2^*)$ (see e.g. \citet{bertsekas2015convex} Proposition 3.2.6).
\end{proof}

\subsection{Pareto Optimality}

\begin{proof}[Counter Examples for PBFP Pareto Optimality]
See Example \ref{example_pareto_pbfp} in the Extra Examples section.
\end{proof}

\begin{proof}[Counter Example to AEF Buyer-Only Pareto Optimality]
See Example \ref{example_pareto_aef} in the Extra Examples section.
\end{proof}

\begin{proof}[Proof that AEF is Buyer-Protected-Item Pareto Optimal]
  Let $x^*$ be the optimal solution to the standard EG program, and let $x^f$ be the optimal solution to EG with the floor constraint added for some set of buyers $C$ and set of protected items $A$.
  We break the proof into two exhaustive cases.

  Case 1: suppose that the EG objective is the same under $x^*$ and $x^f$. In that case, it must be that $u_i(x_i^*, \delta^*_i) = u_i(x_i^f, \delta^f_i)$ for all buyers $i$, since there is a unique set of equilibrium utilities for EG (this follow from strict convexity of the log function). It then follows that $x^f$ must be Pareto optimal, since $x^*$ is Pareto optimal.

  Case 2: suppose that the EG objective strictly decreases. In that case, let $\lambda$ be the Lagrange multiplier on the floor constraint. Since the floor constraint leads to a strict decrease in objective, we must have that $\lambda > 0$. Now, by optimality and Lagrangian duality, we know that $x^f$ must maximize the objective $\sum_i B_i \log u_i(x_i, \delta_i) + \lambda \sum_{i\in \mathcal C} \sum_{j\in A} x_{ij}$ over the set of supply-feasible allocations $x$. But then it follows immediately that $x^f$ must be buyer-protected-item Pareto optimal, since if there exists an alternative allocation $x$ such that $u_i(x_i, \delta_i) \geq u_i(x^f_i, \delta^f_i)$ and $\sum_{i\in \mathcal C} \sum_{j\in A} x_{ij} \geq \sum_{i\in \mathcal C} \sum_{j\in A} x_{ij}^f$, with strict inequality in at least one of these inequalities, then we could strictly improve the Lagrangian objective, which is a contradiction.

  Since the EG objective must either stay the same or decrease under $x^f$, this shows that buyer-protected-item Pareto optimality holds.
\end{proof}

\subsection{Envy Freeness}
In this section we show that allocation according to EG with PBFP, PIFP, or AEF constraints leads to envy-freeness within each buyer group.
The proofs will all use the following observation:
\begin{fact}
  \label{fact:envy}
  For any pair of buyers $i,i'$ in a market equilibrium (possibly with constraints) we have that if buyer $i$ can afford $(B_i/B_{i'})(x_{i'}+\delta_{i'})$ under their personalized prices $p+\bar p_i$, then $u_i(x_i, \delta_i) \geq (B_i/B_{i'})u_i(x_{i'}, \delta_{i'})$.
\end{fact}
This follows from the fact that buyer $i$ receives something in their demand set along with the homogeneity of $v_i$ (and thereby of $u_i$). If they preferred the other bundle (or a scaled version thereof), then they would want to buy that instead.

We can now easily prove that each allocation approach yields (budget-adjusted) envy freeness within each group. For unconstrained buyers the proof is the same for all the mechanisms: since unconstrained buyers all see the same price vector $p$, they already satisfy the affordability condition in Fact \ref{fact:envy}. We thus restrict each proof to handling the case of buyers that face constraints.
\begin{proposition}
  EG with PBFP constraints yields per-group envy freeness.
\end{proposition}
\begin{proof}
  Consider a pair of buyers $i,i'$ that are constrained to satisfy PBFP. We know that their allocations $x_i,x_{i'}$ are such that $\sum_{j\in A} x_{ij} =  \sum_{j\in B} x_{ij}$ and $\sum_{j\in A} x_{i'j} =  \sum_{j\in B} x_{i'j}$.
  It follows that $\bar p_{i}^\top x_i = \bar p_{i'}^\top x_{i'} = 0$, and moreover $\bar p_{i}^\top x_{i'} = \sum_{j\in A} \lambda x_{i'j} - \sum_{j\in B} \lambda x_{i'j} = 0$, where $\lambda$ is the Lagrange multiplier on the PBFP constraint of buyer $i$, which means that $\bar p_{ij} = \lambda$ for $j\in A$ and $\bar p_{ij} = -\lambda$ for $j\in B$.
  Since the only personalized price that $i$ faces is the Lagrange multiplier on their PBFP constraint, it follows immediately that $i$ can afford $(B_i/B_{i'})x_{i'}$ with leftover budget $B_i/B_{i'}\delta_{i'}$, and thus they have no budget-adjusted envy against $i'$.
\end{proof}

\begin{proposition}
  EG with PIFP constraints yields per-group envy freeness.
\end{proposition}
\begin{proof}
  For any pair of buyers $i,i'\in A$, we have that they face the same price $p_j + \lambda_j$ for each item $j$, where $\lambda_j$ is the Lagrange multiplier of the PIFP constraint. It immediately follows that $i$ can afford $(B_i/B_{i'})x_{i'}$ with leftover budget $B_i/B_{i'}\delta_{i'}$. The same argument holds for pairs of buyers in group $B$.
\end{proof}

\begin{proposition}
  EG with AEF constraints yields per-group envy freeness.
\end{proposition}
\begin{proof}
  For any pair of buyers $i,i'\in \mathcal C$, we have that they face the same price $p_j + \lambda_j$ for each item $j$, where $\lambda_j$ is the Lagrange multiplier of the floor constraint. It immediately follows that $i$ can afford $(B_i/B_{i'})x_{i'}$ with leftover budget $B_i/B_{i'}\delta_{i'}$.
\end{proof}

\section{Strategyproofness in the Large}
\label{app:spl}

First we survey the result of \citet{azevedo2018strategy}. 
This is a result for a general class of mechanisms that specify how agents can report types and how that leads to outcomes (which may include allocations as well as payments).
We will consider a variant of their result covered in the appendix of their paper: we will need a more general outcome space, and semi-anonymous mechanisms. We use $\Delta T$ to denote the set of probability distributions over a given set $T$, and $\bar\Delta T$ to denote the set of distributions with full support. 
We note that this section overloads some notation that we used for Fisher markets; this is necessary in order to stay consistent with the notation of \citet{azevedo2018strategy}, and to aid readability in terms of sticking to conventional notation.

Next we specify the assumption made on the mechanism setting for the setting of \citet{azevedo2018strategy}. Let us call this the \emph{Azevedo-Budish setting}.
In the Azevedo-Budish setting, there is a finite set of types $T$ and a measurable set of outcomes $X_0$ (see the appendix of \citet{azevedo2018strategy} for the measurable outcomes case) that an individual agent may receive.
% In the Fisher market context, 
% $T$ would be the set of possible valuation functions (and possibly a binary variable specifying whether a buyer is constrained or not),
% and $X_0$ would be the set of all possible allocation vectors $\mathbb R_+^m$ that a given buyer may receive.
For each type $t_i\in T$ there's a utility function $u_{t_i}: X\rightarrow [0,1]$, where $X=\Delta X_0$ is the set of probability distributions over outcomes that an agent may receive. 
%In the Fisher market context we can work directly with the convex set of outcomes $X=X_0=\mathbb R_+^m$.
Note that in the Azevedo-Budish setup, any possible payments are included in $X_0$ and $u_{t_i}$, so we do not add an explicit extra variable corresponding to leftover budget.

The sets $T$ and $X_0$ are held fixed for all market sizes. For each market size $n$, where $n$ is the number of agents, there is a set $Y_n \subset (X_0)^n$ of feasible allocations. In our Fisher market setting $Y_n$ is simply the set of supply-feasible allocations given the supply of items for a market of size $n$.

For a sequence of feasibility constraints $\{Y_n\}_n$ we say that a \emph{direct mechanism} is a sequence of allocation functions $\left(\Phi^n : T^n \rightarrow \Delta((X_0)^n) \right)_{n\in N}$ such that for all market sizes $n$ and vectors of type reports $t\in T^n$, $\Phi^n(t)$ is contained in $Y_n$.\footnote{\citet{azevedo2018strategy} also handle indirect mechanisms, but we do not need to consider that case here}

We assume that the mechanism is \emph{semi-anonymous}: the types space is partitioned into groups $T = \cup_{g\in G} T_g$. In our Fisher market setting, the groups will be the constrained and unconstrained buyers.
Each agent of type $t\in T_g$ is restricted to reporting a type in $T_g$. For our constrained market setting, this corresponds to the fact that buyers cannot lie about whether they are constrained, but they may misreport their valuation function.

For each $n$ we assume that we are given some type distribution $\mu \in \Delta T$ resulting e.g. from the underlying distribution over types composed with some, possibly randomized, map from types to report types.
We then need the function $\phi^n: T \times \Delta T \rightarrow X$ which specifies the expected allocation function for each buyer given a reported type and a type distribution. This is defined to be 
\[
  \phi^n(t_i,\mu) = \sum_{t_{-i} \in T^{n-1}} \Phi^n_i(t_i, t_{-i}) \cdot \textrm{Pr}(t_{-i}| \mu).
\]

\begin{definition}

  A direct and semi-anonymous mechanism $\{\Phi^n\}_n$ is strategy-proof in the large (SPL) if, for any semi-anonymous type distribution with full support $\mu \in \bar\Delta T$ and $\epsilon > 0$, there exists $n_0$ such that for all $n \geq n_0$, all $g\in G$, and all $t_i,t_i' \in T_g$,
  \[
    u_{t_i}(\phi^n(t_i,\mu)) \geq   u_{t_i}(\phi^n(t_i',\mu)) - \epsilon,
  \]
  i.e. there exists an $n_0$ such that for every type, the incentive to misreport is at most $\epsilon$.
\end{definition}

\citet{azevedo2018strategy} show the following result for semi-anonymous direct mechanisms:
\begin{theorem}
  Assume that we are in the Azevedo-Budish setting.
  If a mechanism is envy free within each group, then it is SPL. Given a type distribution with full support $\mu \in \bar \Delta T$ and $\epsilon > 0$, there exists $C> 0$ such that for all $g\in G, t_i,t_i' \in T_g$ and $n$, the gain from deviating is bounded above by $C\cdot n^{-1/2 + \epsilon}$.
  \label{thm:azevedo budish}
\end{theorem}

Now we consider the following mechanism setting. Each agent reports a type $(v_i, c)$ specifying a homogeneous, concave, and continuous valuation function $v_i$ and a binary variable $c$ denoting whether they are constrained or not (where they cannot lie about the constrained part).
All agents are assumed to have budget $B_i=1$, and all items are assumed to have supply $n / |T|$ (that is, supply grows linearly in the market size; the $1/|T|$ factor is WLOG. and for convenience).
Then the mechanism computes an EG solution with an AEF constraint $\sum_{i\in \mathcal C}\sum_{j\in P} x_{ij}  \geq L(n, |\mathcal C|)$ (here we parameterize the floor $L$ since we generally want the floor constrain to depend on the market size or the number of constrained buyers).
Each agent then receives the corresponding allocation $x_i$ and pays $1-\delta_i$. If multiple agents reported the same type, then from the perspective of EG-AEF we treat them as a single representative buyer $i$ with budget equal to the number of agents that reported the corresponding type, and split the allocation $x_i$ proportionally.
We call this the \emph{EG-AEF mechanism}.

Assume that there is a finite set $T$ of possible valuation functions $v$. Clearly the outcome space $\mathbb R_+^m$ is measurable, and the feasible set of allocation for market size $n$ is the set of supply-feasible allocations and leftover budgets. Moreover, the mechanism is also obviously semi-anonymous: the EG-AEF program is invariant to permutation within the constrained and unconstrained buyer groups.
Now, let $\mu\in \bar\Delta T$ be the full-support distribution of type reports, potentially resulting from the underlying distribution over types composed with the strategy used by each type.
In terms of the Azevedo-Budish setup, we have that $X_0 = \mathbb R^{m+1}_+$ is the set of allocations vectors and leftover budget pairs $(x_i,\delta_i)$ for a given buyer $i$.

\citet{azevedo2018strategy} assume that for each utility function $u_{t_i}$, $u_{t_i}(x) \in [0,1]$ for all $x\in X_0$. But this clearly cannot hold when $X_0=\mathbb R_+^{m+1}$ and our valuation functions $v$ are homogeneous, since if $v_i(x_i)>0$ for any $(x_i,\delta_i)\in X_0$, then by homogeneity we can make $v_i(\alpha x_i)$ arbitrarily large by choosing a sufficiently-large $\alpha$, thus violating the utility lying in $[0,1]$.
To alleviate this fact, consider a modified EG-AEF mechanism: we pick some large cutoff $k$ such that $\mu_{t_i} > 1/k$, and then if the allocation under the reported type vector results in an allocation $x$ such that for some agent $i$ and item $j$ we have $x_{ij} > k$ (after splitting proportionally among all agents that reported this type), then we throw away all the items and set $x_{ij} = 0, \delta = 0$ for all $i,j$.
This modified EG-AEF mechanism is clearly still envy free: when we use an EG-AEF allocation we know that there is no envy within each group, and there's clearly no envy if no agent receives any allocation whatsoever.
It follows that for any full-support type distribution $\mu$, the modified EG-AEF mechanism satisfies the conditions of \cref{thm:azevedo budish} and is thus SPL.

Now we can use this result to show that EG-AEF itself is SPL as a mechanism.
\begin{theorem}
  The EG-AEF mechanism is SPL. Moreover, consider any distribution $\mu$ over types $(v_i, c)$ with full support and an $\epsilon > 0$. Then there exists $C> 0$ such that for all constrained or unconstrained buyers, true valuation functions $v$, alternative valuation function reports $v'$,  and market sizes $n$, the gain from reporting $v'$ rather than $v$ is bounded above by $C\cdot n^{-1/2 + \epsilon}$.
\end{theorem}
\begin{proof}
  We already saw above that the modified EG-AEF mechanism is SPL. Now we show that the EG-AEF mechanism and the modified EG-AEF mechanism achieve the same utility in an asymptotic sense.

  Consider a type distribution $\mu$ with full support. Now consider a market size $n$ and a sampled set of buyers $t = (v_1,c_1),\ldots, (v_n,c_n)$. If the sampled set of buyers are such that the resulting allocation satisfies $x_{ij} > k$ for some $i$, then the payoffs differ between EG-AEF and modified EG-AEF.
  Consider some particular type $(v_i,c_i)$ in this scenario: under modified EG-AEF $i$ gets utility zero.
  We can upper bound the utility that $i$ gets under EG-AEF by the utility that they get from receiving \emph{all} the items,  i.e. $x_i = (n/|T|)\vec{1}$ and set $\delta_i=1$, which is $u_i(x_i,\delta_i) = v_i(x_i) + 1 = (n/|T|)v_i(\vec{1}) + 1$ by homogeneity.
  Let $\bar v = \max_{(v,c)\in T} v(\vec{1}) / |T|$, in which case we upper bound this by $n\bar v + 1$.
  So, in the worst case, a given type may gain $n \bar v + 1$ from misreporting in the case where we sample a set of types $t$ such that EG-AEF and modified EG-AEF differ.

  Let $x^*, \delta^*,x^{mod},\delta^{mod}$ be the solutions from EG-AEF and modified EG-AEF, where we suppress the dependence on the reported types and market size.

  Now fix an $\epsilon >0$ and a type distribution with full support $\mu$. It follows that for any type $(v,c)$, they can gain at most 
  \[
    C\cdot n^{-1/2 + \epsilon} + (n\bar v + 1) \textrm{Pr}(x^* \ne x^{mod}|\mu).
  \]
  Thus, in order to show that EG-AEF is SPL, we need to show that $\textrm{Pr}(x^* \ne x^{mod}|\mu)$ gets small at a fast enough rate.
  In fact, we will show that this rate is exponential.

  Suppose an agent with reported type $t_i$ receives more than $k$ units of some item $j$. In that case, we must have that in the corresponding EG-AEF allocation, the representative buyer $i$ received $k$ times the number of agents that reported type $t_i$ of that good, let number of such reports be $\tau$. First, this implies that the supply of the good is at least $k$, so this only happens when $n \geq k$. 
  In the worst case, all $n$ units of item $j$ are allocated to the representative buyer $i$, in which case we need $n / \tau \geq k$. This implies that $\tau \leq n/k$.
  Since the type reports of the other agents are distributed iid. according to $\mu$, this means that we must sample at most $n/k - 1$ reports of type $t_i$ for $x^*\ne x^{mod}$ to occur.
  The expected number of reports of type $t_i$ is $n\cdot \mu_{t_i}$.

  By Hoeffding's inequality, we now have that the probability of sampling $n/k-1$ reports of type $t_i$ or fewer is upper bounded by $2\exp(-(n\cdot \mu_{t_i} - n/k + 1)^2 / n)$.
  Since we chose $k$ such that $\mu_{t_i} > 1/k$, this grows small exponentially fast in $n$.
\end{proof}

The exact same proof goes through for PIFP and PBFP as well, because the result did not rely on any structure in the AEF constraint apart from the fact that for any pair of buyers $i,i'$ from the same group, we have that they do not envy each other.

\section{Additional Examples}

\begin{table*}[t]%
\begin{tiny}
  \centering
  \subfloat{ \begin{tabular}{lcr}
    \toprule
    V &   &   \\
    \midrule
    Buyer & Item $C$ & Item $U$ \\
    $A$ & 2 & 2 \\ 
    $B$ & .1 & 3 \\ 
    \bottomrule \\
  \end{tabular}}%
  \qquad
  \subfloat{ \begin{tabular}{lcrrr}
    \toprule
    $X^*$ &   &  & & \\
    \midrule
    Buyer & Item $C$ & Item $U$ & $B_i-px_i$ & $u_i$\\
    $A$ & 1 & 0 & 0 & 2\\ 
    $B$ & 0 & 1 & 0 & 3 \\ 
    \bottomrule \\
  \end{tabular}}
  \subfloat{ \begin{tabular}{lcrrr}
    \toprule
    $X^{PIFP}$ &   &  & * \\
    \midrule
    Buyer & Item $C$ & Item $U$ & $B_i-px_i$ & $u_i$\\
    $A$ & .5 & .258 & 0 & 1.52 \\ 
    $B$ & .5 & .742 & 0 & 2.28 \\  
    \bottomrule \\
  \end{tabular}}
  \caption{An example of a set of valuations with budgets set to $1$ where imposing PIFP constraints on a subset of items leads to both originally disadvantaged ($B$) and originally advantaged ($A$) buyers being worse off.}%
  \label{example_everyoneworse_pifp}%
  \end{tiny}
\end{table*}

\begin{table*}[t]%
\begin{tiny}
  \centering
  \subfloat{ \begin{tabular}{lcr}
    \toprule
    V &   &   \\
    \midrule
    Buyer & Item $A$ & Item $B$ \\
    C & 2 & 1 \\ 
    U & 1 & 2 \\ 
    \bottomrule \\
      &   &   
  \end{tabular}}%
  \qquad
  \subfloat{ \begin{tabular}{lcrr}
    \toprule
    $X^*$ &   &  &  \\
    \midrule
    Buyer & Item $A$ & Item $B$ & $B_i-px_i$ \\
    C & 1 & 0 & 0 \\ 
    U & 0 & 1 & 0 \\ 
    \bottomrule \\
    Price & 1 & 1 & -
  \end{tabular}}
  \subfloat{ \begin{tabular}{lcrr}
    \toprule
    $X^{PBFP}$ &   &  &  \\
    \midrule
    Buyer & Item $A$ & Item $B$ & $B_i-px_i$ \\
    C & .5 & .5 & 0 \\ 
    U & .5 & .5 & 0 \\  
    \bottomrule \\
    Price & .66 & 1.33 & -
  \end{tabular}}
  \caption{An example of a set of valuations with budgets set to $1$ where imposing PBFP constraints leads to both parity constrained (C) and unconstrained (U) buyers being worse off.}%
  \label{example_everyoneworse_pbfp}%
  \end{tiny}
\end{table*}

\begin{table*}[t]%
\begin{tiny}
  \centering
  \subfloat{ \begin{tabular}{lcr}
    \toprule
    V &   &   \\
    \midrule
    Buyer & Item $A$ & Item $B$ \\
    C & 2 & 2 \\ 
    U & 0 & 2 \\ 
\bottomrule 

  \end{tabular}}%
  \qquad
  \subfloat{ \begin{tabular}{lcrr}
    \toprule
    $X^*$ &   &  &  \\
    \midrule
    Buyer & Item $A$ & Item $B$ & $B_i-px_i$ \\
    C & 1 & 0 & 0 \\ 
    U & 0 & 1 & 0 \\ 
    \bottomrule 
  \end{tabular}}
  \subfloat{ \begin{tabular}{lcrr}
    \toprule
    $X^{PBFP}$ &   &  &  \\
    \midrule
    Buyer & Item $A$ & Item $B$ & $B_i-px_i$ \\
    C & .5 & .5 & 0 \\ 
    U & .5 & .5 & 0 \\  
    \bottomrule 
  \end{tabular}}
  \caption{An example of a set of valuations with budgets set to $1$ where imposing PBFP constraints leads to an outcome that is Pareto suboptimal both in the buyer only case and the buyer-protected item case (here $B$ is the originally disadvantaged group). $X^{PBFP}$ is Pareto dominated in both cases by any allocation which transfers some item $A$ to buyer $C$.}%
  \label{example_pareto_pbfp}%
  \end{tiny}
\end{table*}

\begin{table*}[t]%
\begin{tiny}
  \centering
  \subfloat{ \begin{tabular}{lcr}
    \toprule
    V &   &   \\
    \midrule
    Buyer & Item $C$ & Item $U$ \\
    A & 0 & 2 \\ 
    B & 2 & 0 \\ 
\bottomrule 

  \end{tabular}}%
  \qquad
  \subfloat{ \begin{tabular}{lcrr}
    \toprule
    $X^*$ &   &  &  \\
    \midrule
    Buyer & Item $C$ & Item $U$ & $B_i-px_i$ \\
    A & 1 & 0 & 0 \\ 
    B & 0 & 1 & 0 \\ 
    \bottomrule 
  \end{tabular}}
  \subfloat{ \begin{tabular}{lcrr}
    \toprule
    $X^{AEF}$ &   &  &  \\
    \midrule
    Buyer & Item $C$ & Item $U$ & $B_i-px_i$ \\
    A & .5 & 1 & 0 \\ 
    B & .5 & 0 & 0 \\  
    \bottomrule 
  \end{tabular}}
  \caption{An example of a set of valuations with budgets set to $1$ where imposing AEF constraints of $x_{AC} \geq .5$ leads to buyer Pareto suboptimal outcomes.}%
  \label{example_pareto_aef}%
  \end{tiny}
\end{table*}

\end{document}